\documentclass[11pt,leqno]{article}
\usepackage{ltexpprt}

\usepackage[a4paper, top=4.2cm, left=3.5cm, right=3.5cm, centering]{geometry}
\usepackage{xspace}
\usepackage{upgreek}
\usepackage{graphicx}
\usepackage[stable]{footmisc}
\usepackage{braket}

\usepackage{amsmath, amsfonts, amssymb}

\usepackage[backend=bibtex,
firstinits,style=numeric-comp,maxnames=99]{biblatex} 

\DeclareFieldFormat{titlecase}{\MakeTitleCase{#1}}

\newrobustcmd{\MakeTitleCase}[1]{%
  \ifthenelse{\ifcurrentfield{booktitle}\OR\ifcurrentfield{booksubtitle}%
    \OR\ifcurrentfield{maintitle}\OR\ifcurrentfield{mainsubtitle}%
    \OR\ifcurrentfield{journaltitle}\OR\ifcurrentfield{journalsubtitle}%
    \OR\ifcurrentfield{issuetitle}\OR\ifcurrentfield{issuesubtitle}%
    \OR\ifentrytype{REMOVEIFEXCLUDEBOOKbook}\OR\ifentrytype{mvbook}\OR\ifentrytype{bookinbook}%
    \OR\ifentrytype{booklet}\OR\ifentrytype{suppbook}%
    \OR\ifentrytype{collection}\OR\ifentrytype{mvcollection}%
    \OR\ifentrytype{suppcollection}\OR\ifentrytype{manual}%
    \OR\ifentrytype{periodical}\OR\ifentrytype{suppperiodical}%
    \OR\ifentrytype{proceedings}\OR\ifentrytype{mvproceedings}%
    \OR\ifentrytype{reference}\OR\ifentrytype{mvreference}%
    \OR\ifentrytype{report}\OR\ifentrytype{thesis}}
    {#1}
    {\MakeSentenceCase{#1}}}

\bibliography{longnames,bib-latest}

\renewcommand{\geq}{\geqslant}
\renewcommand{\leq}{\leqslant}
\renewcommand{\epsilon}{\varepsilon}
\renewcommand{\Delta}{\Updelta}

\newtheorem{problem}{Problem}[section]

\newcommand{\calU}{\ensuremath{{\mathcal{U}}}}

\newcommand{\Hamarray}{\ensuremath{\textup{HamArray}}}

\newcommand{\Prob}{\ensuremath{\textup{Pr}}}
\newcommand{\expected}[1]{\ensuremath{\mathbb{E}\left[#1\right]}}
\newcommand{\Ham}[1]{\ensuremath{\textup{Ham}(#1)}}
\newcommand{\ball}[1]{\ensuremath{\textup{Ball}(#1)}}

\newcommand{\vsum}{\ensuremath{\textup{Sum}}}

\newcommand{\D}{{\ensuremath{D}}} 
\newcommand{\F}{{\ensuremath{F}}} 
\renewcommand{\S}{{\ensuremath{S}}} 
\newcommand{\U}{{\ensuremath{U}}} 
\newcommand{\n}{{\ensuremath{n}}} 

\newcommand{\arrive}{\ensuremath{\textup{arrive}}}
\newcommand{\outside}{\ensuremath{^\textup{c}}}
\newcommand{\fix}{\ensuremath{_\textup{fix}}}

\newcommand{\Cbig}{\ensuremath{\widetilde{C}}}

\newcommand{\cbig}{\ensuremath{\widetilde{c}}}

\newcommand{\ELL}{\ensuremath{\mu}}
\newcommand{\ps}{\ensuremath{\rho}}
\newcommand{\psymb}{\ensuremath{\ast}}

\newcommand{\tsymb}{\ensuremath{\diamond}}

\newcommand{\insertdiagram}[1]{\includegraphics[scale=0.82]{figures/#1}}
\newcommand{\Patrascu}{P{\v a}tra{\c s}cu\xspace}

\begin{document}

\pagestyle{plain}

\title{\LARGE Tight Cell-Probe Bounds for Online \\Hamming Distance Computation\thanks{Supported by EPSRC.}}
\author{Rapha\"el Clifford\thanks{Department of Computer Science, University of Bristol, Bristol, BS8~1UB, U.K.}
\and
Markus Jalsenius\footnotemark[2]
\and
Benjamin Sach\thanks{Department of Computer Science, University of Warwick, Coventry, CV4~7AL, U.K.}
}
\date{}

\maketitle


\begin{abstract}
We show tight bounds for online Hamming distance computation in the cell-probe model with word size $w$. The task is to output the Hamming distance between a fixed string of length $n$ and the last $n$ symbols of a stream. We give a lower bound of $\Omega\left(\frac{\delta}{w}\log n\right)$ time on average per output, where $\delta$ is the number of bits needed to represent an input symbol. We argue that this bound is tight within the model. The lower bound holds under randomisation and amortisation.
\end{abstract}


\section{Introduction}

We consider the complexity of computing the Hamming distance. The question of how to compute the Hamming distance efficiently has a rich literature, spanning many of the most important fields in computer science. Within the theory community, communication complexity based lower bounds and streaming model upper bounds for the Hamming distance problem have been the subject of particularly intense study~\cite{CDIM:03,Woodruff:04,HSZZ:06,JKS:08,BCRT:10,Chakrabarti:11}. This previous work has however almost exclusively focussed on providing resource bounds (either in terms of space or bits of communication) for computing approximate answers. 

 We give the first  time complexity lower bounds for exact Hamming distance computation in an online or streaming context.  Our results are in the cell-probe model where we also provide matching upper bounds.

\begin{problem}[Online Hamming distance]
    For a fixed string $F$ of length $n$, we consider a stream in which symbols arrive one at a time. For each arriving symbol, before the next symbol arrives, we output the Hamming distance between $F$ and the last $n$ symbols of the stream.
\end{problem}

We show that there are instances of this problem for which any algorithm solving it will require $\Omega\left(\frac{\delta}{w}\log n\right)$ time on average per output, where   $\delta$ is the number of bits needed to represent an input symbol and $w$ is the number of bits per cell in the cell-probe model. Lower bounds in the cell-probe model also hold for the popular word-RAM model in which much of today's algorithms are given.  The full statement and the main result of this paper is given in Theorem~\ref{thm:lower}.

\newcommand{\statementLower}{
    In the cell-probe model with $w$ bits per cell there exist instances of the online Hamming distance problem such that the expected amortised time per arriving symbol is $\Omega\left(\frac{\delta}{w}\log n\right)$.
}
\begin{theorem}
    \label{thm:lower}
    \statementLower
\end{theorem}

Where $\delta = w$, for example we have an  $\Omega(\log{n})$ lower bound.  Despite the relatively modest appearance of this bound, we prove that it is in fact tight within the cell-probe model. Furthermore, we argue that it is likely to be the best bound available for this problem without a significant breakthrough in computational complexity.

The first cell-probe lower bounds in this online model were given recently for the problems of online convolution and multiplication~\cite{CJ:2011}.  This work introduced the use of the information transfer method~\cite{PD2004:Partial-sums} from the world of data structure lower bounds to this class of online or streaming problems.  Information transfer captures the amount of information that is transferred from the operations in one time interval  to the next and hence the minimum number of cells that must be read to compute a new set of answers.

Our key innovation is a carefully designed fixed string $F$ which together with a random distribution over a subset of possible input streams provide a lower bound for the information transfer between successive time intervals.  The string $F$ is derived by a sequence of transformations. These start with binary cyclic codes and go via binary vectors with many distinct subsums and an intermediate string to finally arrive at $F$ itself. The use of such a purposefully designed input departs from the most closely related previous work~\cite{CJ:2011} and also from much of the lower bound literature where simple uniform distributions over the whole input space often suffice. It also however necessarily creates a number of challenging technical hurdles which we overcome.

%

The central fact that had previously enabled a lower bound to be proven for the online convolution problem was that the inner product  between a vector and successive suffixes of the stream reveals a lot of information about the history of the stream.  Establishing a similar result for online Hamming distance problem appears, however, to be considerably more challenging for a number of reasons.  The first and most obvious is that the amount of information one gains by comparing whether two potentially large symbols are equal is at most one bit, as opposed to $O(\log{n})$ bits for multiplication.  The second is that a particularly simple worst case string could be found for the convolution problem which greatly eased the resulting analysis. We have not been able to find such a simple fixed string for the Hamming distance problem and our proof of the existence of a hard instance is non-constructive and involves a number of new insights, combining ideas from coding theory and additive combinatorics.

The bounds we give are also tight within the cell-probe model. This can be seen by application of an existing general reduction from online to offline pattern matching problems~\cite{CEPP:2011}.  In this previous work it was shown that any offline algorithm for Hamming distance computation can be converted to an online one with at most an $O(\log{n})$ factor overhead. For details of these reductions we refer the reader to the original paper.  In our case, the same approach also allows us to directly convert any cell-probe algorithm from an offline to online setting. An offline cell-probe algorithm for Hamming distance could first read the whole input, then compute the answers and finally output them. This takes $O{\left(\frac{\delta}{w} n\right)}$ cell probes. We can therefore derive an online cell-probe algorithm which takes only $O{\left(\frac{\delta}{w}\log n\right)}$ probes per output, matching the new lower bound we give.  We state the final result in Corollary~\ref{cor:final}.

\begin{corollary}\label{cor:final}
 The expected amortised cell-probe complexity of the online Hamming distance problem is $\Theta(\frac{\delta}{w}\log n)$.
\end{corollary}

One consequence of our results is the first strict separation between the complexity of exact and inexact pattern matching.  Online exact matching can be solved in constant time~\cite{Galil:1981} per new input symbol and our new lower bound proves for the first time that this is not possible for Hamming  distance.   Previous results had only shown such a separation for algorithms which made use of fast convolution computation~\cite{CJ:2011}.  Our new lower bound immediately opens the interesting and as yet unresolved question of how to show the same separation for other distance measures such as, for example, edit distance.  The edit distance between two strings is widely assumed to be harder to compute than the Hamming distance but this has yet to be proven.

Our lower bound also implies a matching lower bound for any problem that Hamming distance can be reduced to.  The most straightforward of these is online $L_1$ distance computation, where the task is to output the $L_1$ distance between a fixed vector of integers and the last $n$ numbers in the stream. A suitable reduction was shown in~\cite{LP:2008}. The expected amortised cell probe complexity for the online $L_1$ distance problem is therefore also $\Theta{\left(\frac{\delta}{w}\log n\right)}$ per new output.



\subsection{Technical contributions}

The use of information transfer to provide time lower bounds is not new,  originating from~\cite{PD2004:Partial-sums}.  However, applying the method to our problem has required a number of new insights and technical innovations. Perhaps the most surprising of these is a new relationship between the Hamming distance, vector sums and constant weight binary cyclic codes.

When computing the Hamming distance there is a balance between the number of symbols being used and the length of the strings. For large alphabets and short strings, one would expect a typical outputted Hamming distance to be close to the length of the string on random inputs and therefore to provide very little information. This suggests that the length of the strings must be sufficiently long in relation to the alphabet size to ensure that the entropy of the outputs is large (a property required by the information transfer method). On a closer look, it is not immediately obvious that large entropy can be obtained unless the fixed string that is being compared to the input stream is \emph{exponentially} larger than the alphabet size. This potentially poses another problem for the information transfer method, namely that the word size $w$ would be much larger than $\delta$ (the number of bits needed to represent a symbol), making a $\log n$ lower bound impossible to achieve.

Our main technical contribution is to show that fixed strings of length only polynomial in the size of the alphabet exist which provide sufficiently high entropy outputs.  Such strings, when combined with a suitable input distribution maximising the number of distinct Hamming distance output arrays, give us the overall lower bound.  We design a fixed string $F$ with this desirable probably in such a way that there is a one-to-one mapping between  many of the different possible input streams and the outputted Hamming distances. This in turn implies large entropy. The construction of $F$ is non-trivial and we break it into smaller building blocks, reducing our problem to a purely combinatorial question relating to vectors sums. That is, given a relatively small set $V$ of vectors of length $m$, how many distinct vector sums can be obtained by choosing $m$ vectors from $V$ and adding them (element-wise). We show that even if we are restricted to picking vectors only from subsets of $V$, there exists a $V$ such that the number of distinct vector sums is  $m^{\Omega(m)}$. We believe this result is interesting in its own right. Our proof for the combinatorial problem is non-constructive and probabilistic, using constant weight cyclic binary codes to prove that there is a positive probability of the existence of a set $V$ with the desired property.


\subsection{The cell-probe model}

 Our bounds hold in a particularly strong computational model, the \emph{cell-probe model}, introduced originally by Minsky and Papert~\cite{MP:1969} in a different context and then subsequently by Fredman~\cite{Fredman:1978} and Yao~\cite{Yao1981:Tables}. In this model, there is a separation between the computing unit and the memory, which is external and consists of a set of cells of $w$ bits each. The computing unit cannot remember any information between operations. Computation is free and the cost is measured only in the number of cell reads or writes (cell-probes). This general view makes the model very strong, subsuming for instance the popular word-RAM model. In the word-RAM model certain operations on words, such as addition, subtraction and possibly multiplication take constant time (see for example~\cite{Hagerup:1998} for a detailed introduction). Here a word corresponds to a cell. As is typical, we will require that the cell size $w$ is at least $\log_2 n$ bits. This allows each cell to hold the address of any location in memory.

The generality of the cell-probe model makes it particularly attractive for establishing lower bounds for dynamic data structure problems and many such results have been given in the past couple of decades. The approaches taken had historically been based only on communication complexity arguments and the chronogram technique of Fredman and Saks~\cite{FS1989:chronogram}.
However in 2004, a breakthrough lead by \Patrascu and Demaine gave us the tools to seal the gaps for several data structure problems~\cite{PD2006:Low-Bounds} as well as giving the first $\Omega(\log{n})$ lower bounds. The new technique is based on information theoretic arguments that we also deploy here. \Patrascu and Demaine also presented ideas which allowed them to express more refined lower bounds such as trade-offs between updates and queries of dynamic data structures.  For a list of data structure problems and their lower bounds using these and related techniques, see for example~\cite{Pat2008:Thesis}.  Very recently, a new lower bound of $\Omega\left((\log{n}/\log{\log{n}})^2\right)$ was given for the cell-probe complexity of performing queries in the dynamic range counting problem~\cite{Larsen:2012}. This result  holds under the natural assumptions of $\Theta(\log{n})$ size words and polylogarithmic time updates and is another exciting breakthrough in the field of cell-probe complexity.

\subsection{Barriers to improving our bounds}\label{sec:previous}

The cell probe bound we give is tight within the model but  still distant from the time complexity of the fastest known RAM algorithms. For the online Hamming distance problem, the best known complexity is $O(\sqrt{n\log{n}})$ time per arriving symbol~\cite{CEPP:2011}. It is therefore tempting to wonder if better upper or lower bounds can be found by some other not yet discovered method. This however appears challenging for at least two reasons.  First, a higher lower bound than $\Omega(\log{n})$ immediately implies a superlinear offline lower bound for Hamming distance computation by the online to offline reduction of~\cite{CEPP:2011}.  This would be a truly remarkable breakthrough in the field of computational complexity as no such offline lower bound is known even for the canonical NP-complete problem SAT.  On the other hand, an improvement of the upper bound for Hamming  distance computation to meet our lower bound would also have significant implications. A reduction that is now regarded as folklore (see appendix) tells us that any $O(f(n))$ algorithm for Hamming distance computation, assuming pattern of length $n$ and text of length $2n$, implies an $O(f(n^2))$ algorithm for multiplying $n\times n$ binary matrices over the integers.   Therefore an $O(\log{n})$ time online Hamming distance algorithm would imply an $O(n\log{n})$  offline Hamming distance algorithm, which would in turn imply an $O(n^2\log{n})$ time algorithm for binary matrix multiplication.  Although such a result would arguably be less shocking than a proof of a superlinear lower bound for Hamming distance computation, it would nonetheless be a significant breakthrough in the complexity of a classic and much studied problem.

\subsection{Previous results for exact Hamming distance computation}

Almost all previous algorithmic work for exact Hamming distance computation has considered the problem in an offline setting.  Given a pattern, $P$ and a text, $T$, the best current deterministic upper bound for offline Hamming distance computation is an $O(|T|\sqrt{|P|\log{|P|}})$ time algorithm based on convolutions~\cite{Abrahamson:1987, Kosaraju:1987}. In \cite{Karloff:1993} a randomised algorithm was given that takes $O((|T|/{\epsilon}^2)\log^2{|P|})$ time which was subsequently modified in~\cite{Indyk:1998} to $O((|T|/{{\epsilon}^3}) \log{|P|})$. Particular interest has also been paid to a bounded version of this problem called the $k$-mismatch problem.  Here a bound $k$ is given and we need only report the Hamming distance if it is less than or equal to $k$. In \cite{LV:1986a}, an $O(|T|k)$ algorithm was given that is not convolution based and uses $O(1)$ time lowest common ancestor (LCA) operations on the suffix tree of $P$ and $T$. This was then  improved to $O(|T|\sqrt{k\log{k}})$ time by a method that combines LCA queries, filtering and convolutions~\cite{ALP:2004}.

\section{Proof overview}

Let us first introduce some basic notation which we will use throughout. For positive integer $n$, we define $[n]=\{0,\dots,n-1\}$. For a string $S$ of length $n$ and $i,j\in[n]$, we write $S[i]$ to denote the symbol at position $i$, and where $j\geq i$, $S[i,j]$ denotes the $(j-i+1)$~length substring of $S$ starting at position $i$. The string $S_1S_2$ denotes the concatenation of strings $S_1$ and $S_2$. We say that $S$ is over the alphabet $\Sigma$ if $S[i]\in \Sigma$ for all $i\in[n]$. The Hamming distance between two strings $S$ and $S'$ of the same length $n$, denoted $\Ham{S,S'}$, is the number of positions $i\in[n]$ for which $S[i]\neq S'[i]$.

The \emph{online Hamming distance problem} is parameterised by a positive integers $n$ and $\delta$ and a string $\F\in\Sigma^n$ where $|\Sigma| \leq 2^\delta$. The parameter $\delta$ therefore denotes the smallest possible number of bits needed to represent a symbol in the alphabet, $\Sigma$. The problem is to maintain a string $\S\in\Sigma^n$ subject to an operation $\arrive(x)$ which takes a symbol $x\in\Sigma$, modifies $S$ by removing the leftmost symbol $\S[0]$ and appending $x$ to right of the rightmost symbol $\S[n-1]$, and then returns the Hamming distance $\Ham{\F,\S}$ between $\F$ and the updated $\S$. We refer to the operation $\arrive(x)$ as the \emph{arrival} of symbol $x$.




In order to prove Theorem~\ref{thm:lower} we will consider a carefully chosen string $\F$ with a random sequence of $n$ arriving symbols and show that the expected running time over these arrivals is $\Omega{(\frac{\delta}{w}n\log n)}$.
We let the $n$~length string $\U\in\Sigma^n$ contain the $n$ arriving symbols of the update sequence and we use $t\in[n]$ to denote the time, where the operation $\arrive(\U[t])$ is said to occur at time $t$.
When referring to updates of $\U$ that take place outside some time interval $[t_0,t_1]$, we use the notation $\U[t_0,t_1]\outside$ to denote the sequence $\U[0]\cdots \U[t_0-1]\U[t_1+1]\cdots \U[n-1]$.
The choice of $\F$ and distribution of updates $\U$, which we defer to Section~\ref{sec:hard-instance}, is the most challenging aspect of this work.

We let the $n$~length array $\D\in[n+1]^n$ denote the Hamming distances outputted during the update sequence $\U$ such that, for $t\in[n]$, $\D[t]=\Ham{\F,\S}$, where $\S$ has just been updated by the arrival of $\U[t]$. 

Following the overall approach of Demaine and \Patrascu~\cite{PD2004:Partial-sums} we will consider adjacent time intervals and study the \emph{information} that is transferred from the operations in one interval to the next. Let $t_0,t_1,t_2\in[n]$ such that $t_0\leq t_1<t_2$ and consider any algorithm solving the online Hamming distance problem.
We define the \emph{information transfer}, denoted $IT(t_0,t_1,t_2)$, to be the set of memory cells $c$ such that $c$ is written during the first interval $[t_0,t_1]$, read at some time $t$ in the subsequent interval $[t_1+1,t_2]$ and not written during $[t_1+1,t]$. Hence a cell that is overwritten in the second interval before being read, is not in the information transfer. The cells of the information transfer contain all the information about the arriving symbols in the first interval that the algorithm uses in order to correctly output the Hamming distances $\D[t_1+1,t_2]$ during the second interval. This fact is captured in the following lemma, which was stated with small notational differences as Lemma~3.2 in~\cite{Pat2008:Thesis}. For completeness we include a full proof.
The overall idea of the proof is to describe an encoding of the information transfer such that any algorithm running on the $n$ arrivals in $\U$, where the symbols outside the first interval, $\U[t_0,t_1]\outside$, are fixed to some known values $\U\fix[t_0,t_1]\outside$, can correctly output the Hamming distances $\D[t_1+1,t_2]$ by using the values $\U\fix[t_0,t_1]\outside$ and decoding the information transfer.

\newcommand{\statementEntropyUpper}{
    The entropy
    \begin{align*}
        &H\big(\D[t_1+1,t_2] \;\big|\; \U[t_0,t_1]\outside = \U\fix[t_0,t_1]\outside \big)\\
         &\leq w + 2w\cdot \mathbb{E}\big[|IT(t_0,t_1,t_2)| \;\big|\; \U[t_0,t_1]\outside = \U\fix[t_0,t_1]\outside\big]\,.
    \end{align*}
}
\begin{lemma}[Lemma~3.2 of \cite{Pat2008:Thesis}]
    \label{lem:entropy-upper}
    \statementEntropyUpper
\end{lemma}
\begin{proof}
    The average length of any encoding of $\D[t_1+1,t_2]$ (conditioned on $\D\fix[t_0,t_1]\outside$) is an upper bound on its entropy. We use the information transfer as an encoding in the following way. For every cell $c$ in the information transfer $IT(t_0,t_1,t_2)$, we store the address of $c$, which takes at most $w$ bits under the assumption that the cell size can hold the address of every cell, and we store the contents of $c$, which is a cell of $w$ bits. In total this requires $2w\cdot |IT(t_0,t_1,t_2)|$ bits which are stored consecutively as an array of cells in the memory. In addition we store the size of the information transfer, $|IT(t_0,t_1,t_2)|$, so that any algorithm decoding the stored information knows where the end of the array is. Storing the size of the information transfer requires $w$ bits, thus the average length of the encoding is $w + 2w\cdot \mathbb{E}\big[|IT(t_0,t_1,t_2)| \;\big|\; \U[t_0,t_1]\outside = \U\fix[t_0,t_1]\outside\big]$.

    In order to prove that the described encoding is valid, we describe how to decode the stored information. We do this by simulating the algorithm. First we simulate the algorithm from time 0 to $t_0-1$. We have no problem doing so since all necessary information is available in $\D\fix[t_0,t_1]\outside$, which we know. We then skip from time $t_0$ to $t_1$ and resume simulating the algorithm from time $t_1+1$ to $t_2$. In this interval the algorithm outputs the Hamming distances in $\D[t_1+1,t_2]$. In order to correctly do so, the algorithm might need information about the symbols that arrived during the interval $[t_0,t_1]$. This information is only available through the encoding described above. When simulating the algorithm, for each cell $c$ we read, we check if the address of $c$ is contained in the list of addresses that was stored. If so, we obtain the contents of $c$ by reading its stored value. Each time we write to a cell whose address is in the list of stored addresses, we remove it from the stored list, or blank it out. Note that every cell we read whose address is not in the stored list contains a value that was written last either before time $t_0$ or after time~$t_1$. Hence its value is known to us.
    \hfill $\square$
\end{proof}

While an encoding of the information transfer provides an upper bound on the entropy of the outputs in the interval $[t_1+1,t_2]$, the question is how much information about the symbols arriving in $[t_0,t_1]$ \emph{needs} to be communicated from $[t_0,t_1]$ to $[t_1+1,t_2]$. We answer this question in the next lemma by providing a lower bound on the entropy. The lemma is key to this paper and its proof is given in Section~\ref{sec:hard-instance} where we show that there is a string $\F$ such that for a large set of updates~$\U$, the outputs $\D[t_1+1,t_2]$ uniquely specify a constant fraction of the symbols that arrived in $[t_0,t_1]$.

\begin{lemma}
    \label{lem:entropy-lower}
    There exists a string $\F$ and distribution of updates $\U$ such that for any two intervals $[t_0,t_1]$ and $[t_1+1,t_2]$ of the same length~$2^\ell \geq k\sqrt{\n}$ where $k>0$ is a constant, the entropy
    \begin{equation*}
        H\big(\D[t_1+1,t_2] \;\big|\; \U[t_0,t_1]\outside = \U\fix[t_0,t_1]\outside \big)
            \;\in\; \Omega{(\delta\cdot2^\ell)} \,.
    \end{equation*}
\end{lemma}

\noindent We combine Lemmas~\ref{lem:entropy-upper} and~\ref{lem:entropy-lower} in the following corollary. 

\newcommand{\statementIT}{
    There exists a string $\F$ and distribution of $\U$ such that for any two intervals $[t_0,t_1]$ and $[t_1+1,t_2]$ of the same length $2^\ell \geq k \sqrt{\n}$, where $k>0$ is a constant, and any algorithm solving the online Hamming distance problem, 

    \begin{equation*}
        \mathbb{E}\big[|IT(t_0,t_1,t_2)|\big] \;\in\; \Omega{\left(\frac{\delta}{w}\cdot2^\ell\right)} \,.
    \end{equation*}
}
\begin{corollary}
    \label{cor:IT}
    \statementIT
\end{corollary}
\begin{proof}
    For $\U[t_0,t_1]\outside$ fixed to $\U\fix[t_0,t_1]\outside$, by comparing Lemmas~\ref{lem:entropy-upper} and~\ref{lem:entropy-lower} we see that
    \begin{equation*}
        \mathbb{E}\big[|IT(t_0,t_1,t_2)| \;\big|\; \U[t_0,t_1]\outside = \U\fix[t_0,t_1]\outside\big] \;\geq\; \frac{\delta \cdot 2^\ell}{2w}-\frac{1}{2}\,.
    \end{equation*}
    The result follows by taking expectation over $\U[t_0,t_1]\outside$ under the random sequence~$\U$.
    
    \hfill$\square$
\end{proof}

We are now in a position to prove Theorem~\ref{thm:lower}.

\begin{proof}[Proof of Theorem~\ref{thm:lower}] The main idea is to sum the information transfer between many pairs of time intervals and show that over the $n$ arrivals in $U$, a large amount of information must have been transferred. To capture this idea, we conceptually think of a balanced binary tree over the time axis, where the leaves, from left to right, represent the time $t$ from $0$ to $n-1$, respectively. An internal node $v$ is associated with the times $t_0$, $t_1$ and $t_2$ such that the two intervals $[t_0,t_1]$ and $[t_1+1,t_2]$ span the left subtree and the right subtree of $v$, respectively. The information transfer $IT(v)$ associated with $v$ is $IT(t_0,t_1,t_2)$. This idea was introduced in~\cite{PD2004:Partial-sums} in the context of showing a lower bound for the partial sums problem. Here we use a refined version of the technique which we need in order to cope with short intervals; the lower bound on the size of the information transfer is by Corollary~\ref{cor:IT} only guaranteed for sufficiently large intervals.

A crucial property of the information transfer method which was proven in~\cite{PD2004:Partial-sums} is that the cell probes counted in some information transfer $IT(v)$ associated with a node $v$ of the tree are not counted in $IT(v')$ of any other node $v'$. Therefore, using linearity of expectation, we have that the sum over all nodes, $\sum_v \mathbb{E}[IT(v)]$ is a lower bound on the expected number of cell probes over $n$ updates. However, the lower bound on $\mathbb{E}[IT(v)]$ given by Corollary~\ref{cor:IT} is only guaranteed for nodes representing sufficiently large intervals. Fortunately, this includes all nodes in the top $\log \sqrt{n} - O(1)$ levels of the tree. By summing the information transfer over these nodes we have that any algorithm performs $\Omega\big(\frac{\delta}{w}\,n\log n\big)$ expected cell probes over $n$ updates. This concludes the proof of Theorem~\ref{thm:lower}. The remainder of the paper concerns the proof of Lemma~\ref{lem:entropy-lower} upon which the main result relies.

Although we have only shown the existence of probability distributions on the inputs for which we can prove lower bounds on the expected running time of any deterministic algorithm, by Yao's minimax principle~\cite{Yao1977:Minimax} this also immediately implies that for every (randomised) algorithm, there is a worst-case input such that the (expected) running time is equally high.  Therefore our lower bounds hold equally for randomised algorithms as for deterministic ones.
    \hfill$\square$
\end{proof}

\section{The hard instance}\label{sec:hard-instance}

\begin{figure}[t]
    \centering
    \insertdiagram{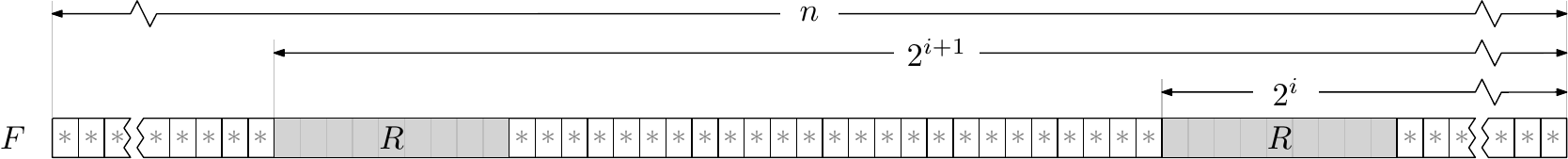}
    \caption{\label{fig:fstring}The string $F$ has a copy of $R$ starting at each position that is a power of two from the end. All other positions have the symbol~\psymb. (In this particular diagram, $n$ is not a power of two.)}
\end{figure}

In this section we discuss the proof of Lemma~\ref{lem:entropy-lower}, the lower bound on the conditional entropy of $\D[t_1+1,t_2]$.
We will describe a string $F$ with the property that the outputs $\D[t_1+1,t_2]$ during the interval $[t_1+1,t_2]$ determine the values of a constant fraction of the symbols in $U[t_0,t_1]$. By picking the update sequence $U$ from a large set of strings we ensure that the entropy of $\D[t_1+1,t_2]$ is large.

The description of the \emph{hard instance}, i.e. the string $F$, is given in two parts. In this section we use a certain string, denoted $R$, to construct $F$. A full description of $R$ itself is given separately in Section~\ref{sec:R} and is where the non-constructive part of the proof lies. For the purpose of constructing $F$ we need $R$ to have a particular property, which is stated in Lemma~\ref{lem:combinatorial} below. The reason why Lemma~\ref{lem:combinatorial} is important will be clear  shortly. First we introduce some notation.

For a string $S_1$ of length $m$ and a string $S_2$ of length $2m$, we write $\Hamarray(S_1,S_2)$ to denote the array of length $m+1$ such that, for $i\in[m+1]$, $\Hamarray(S_1,S_2)[i]=\Ham{S_1,S_2[i,\;i+m-1]}$. That is, $\Hamarray(S_1,S_2)$ is the array of Hamming distances between $S_1$ and every $m$~length substring of $S_2$.

\begin{lemma}
\label{lem:combinatorial}
    For any $r$ there exists a string $R \in [r]^r$ such that $$\log \big|\set{\Hamarray(R,U') \;|\; \textup{$U' \in [r]^{2r}$}}\big| \;\in\; \Omega(r\log r)\,.$$
\end{lemma}

In Section~\ref{sec:R} we describe how
$R$ is partitioned into many smaller substrings containing distinct symbols. By choosing the substrings at random, we show that there is a positive probability of getting a string $R$ with the desired property and hence such an $R$ exists. The proof of Lemma~\ref{lem:combinatorial} will demonstrate an interesting connection between Hamming distances, vector sums and cyclic codes.

For the construction of $F$ we set $r=2^\delta-1$. Recall that $\delta$ is the number of bits needed to represent a symbol of the alphabet. The ``minus one'' ensures that we can reserve one symbol  that does not appear in the alphabet $[r]$ over which $R$ is defined. We use $\psymb$ to denote this symbol.

From now on, let $R$ be a string with the property of Lemma~\ref{lem:combinatorial}.
We define $\calU\subseteq [r]^{2r}$ to be a largest set of strings such that for any two distinct $\U'_{1},\U'_{2}\in\calU$, $\Hamarray(R,\U'_1)\neq \Hamarray(R,\U'_2)$. The set $\calU$ is not unique and is chosen arbitrarily as long as its size is maximised. By Lemma~\ref{lem:combinatorial} we have that the size of $\calU$ is at least $r^{cr}$ for some constant $c$. We will see that setting the update sequence $U$ to consist of strings chosen randomly from $\calU$ yields a sufficiently rich variety of outputted Hamming distances $\D[t_1+1,t_2]$ to make its entropy large.

\subsection{The fixed string $F$}

For the construction of $F$ to work  we need the length $n$ of $F$ to be sufficiently large in comparison to $r$. To avoid unnecessary technicalities, think of $n$ as being at least $r^2$. Observe that with the cell size $w=\log n$ and $n$ being any polynomial in $r$ (i.e. $\delta=\Theta(\log n)$ as $r=2^\delta-1$), our lower bound for the online Hamming distance problem simplifies to $\Omega(\log n)$.

The perhaps easiest way to describe $F$ is to start with an $n$~length string that consists entirely of the symbol $\psymb$, i.e. the string $\{\psymb\}^n$ and then replace $r$~length substrings $\{\psymb\}^r$ with the string $R$ as follows. Refer to Figure~\ref{fig:fstring}. For each suffix of $\{\psymb\}^n$ that has a power-of-two length and is longer than $r$, replace its $r$~length prefix with the string $R$. There will therefore be a logarithmic number of copies of $R$ in $F$, and they all start at power-of-two positions from the end of $F$. 

The benefit of this construction is seen with the following reasoning, which outlines the proof of Lemma~\ref{lem:entropy-lower}.
Randomly pick an $n$~length string $W$ that is the concatenation of $n/(2r)$ $2r$~length strings chosen independently and uniformly at random from $\calU$. Consider an interval $[t_1+1,t_2]$ that has length $2^\ell$ (which we assume is at least a constant times $r$) and let $U$ be the update sequence such that the updates during $[t_0,t_1]$ are induced by $W$ (i.e. $\U[t_0,t_1]=W[t_0,t_1]$) whereas the updates $\U[t_0,t_1]\outside$ are fixed to some arbitrary $\U\fix[t_0,t_1]\outside$ over the alphabet $[r]$. It will be useful to refer to Figure~\ref{fig:doubling} when reading the rest of this section. Here we have drawn the string $F$ with three occurrences of the substring $R$. We have also drawn the text stream under three different alignments with $F$, labelled 1,2 and~3, respectively. The first alignment corresponds to the time $t_1$ where $\U[t_0,t_1]$ has just been fed into the stream. In this particular example, $\U[t_0,t_1]$ consists of eight $2r$ length substrings drawn from $\calU$, labelled $U'_1,\dots,U'_8$. The third alignment corresponds to the time $t_2$ where another $2^\ell$ symbols have been fed into the stream.

\begin{figure*}[t]
    \centering
    \insertdiagram{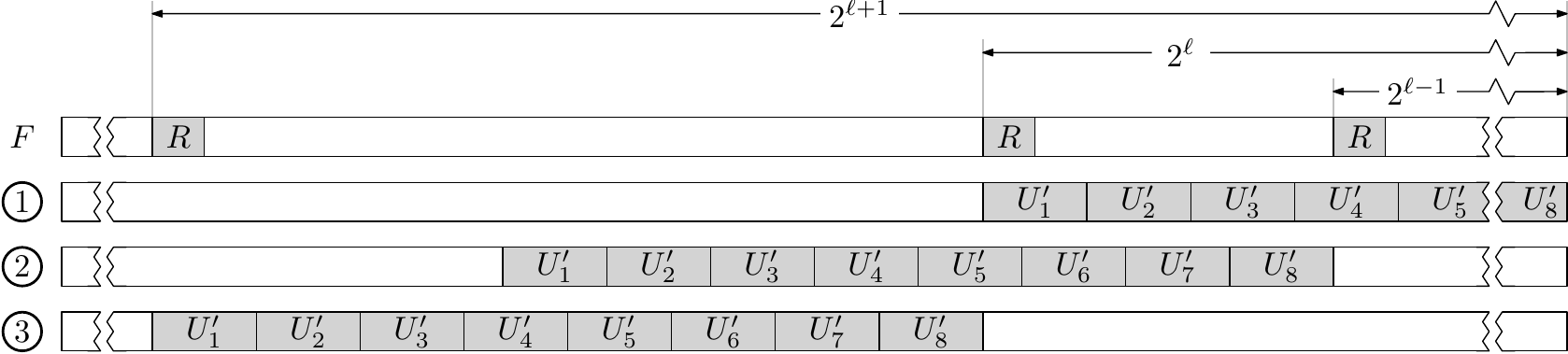}
    \caption{\label{fig:doubling}The string $F$ and three different alignments with the text stream. The first alignment is at time $t_1$ and the third alignment at time $t_2$. The string $\U[t_0,t_1]$ is here the concatenation of $U'_1,\dots,U'_8$.}
\end{figure*}

The contribution to the outputs of $\D[t_1+1,t_2]$ can be split into two parts: those coming from mismatches with $\U[t_0,t_1]$ and those coming from mismatches with $\U[t_0,t_1]\outside$. Since the latter is known to us, we can derive from $\D[t_1+1,t_2]$ the contribution from the former. From the construction of $F$, it is not too difficult to see that for most of the second half of the interval $[t_1+1,t_2]$ (starting with the second alignment of Figure~\ref{fig:doubling}), all the unknown updates of $\U[t_0,t_1]$ are aligned with either $\psymb$ (causing a mismatch) or the occurrence of $R$ at the start of the $2^\ell$~length suffix of $\F$. Around half of the $2r$~length substrings of $\U[t_0,t_1]$ drawn from $\calU$ will in turn slide over this copy of $R$ while every other symbol of $\U[t_0,t_1]$ is aligned with $\psymb$. In the example of Figure~\ref{fig:doubling} we see that between the second and the third alignment, the two substrings $U'_6$ and $U'_7$ will each slide over $R$ while all other symbols of $\U[t_0,t_1]$ are aligned with $\psymb$. It is not difficult to see that if we scale the example to contain many more that eight $2r$ substrings, half of them, minus a constant number, will indeed have the property just described. Each such substring $U'\in\calU$ contributes $\Hamarray(R,U')$ to the outputs and as we reasoned above, $\Hamarray(R,U')$ can therefore be derived from the corresponding substring of $\D[t_1+1,t_2]$.
By definition there is only one string $U'\in\calU$ that can give rise to $\Hamarray(R,U')$, hence we can uniquely identify which string $U'$ of $\calU$ was chosen. In the example of Figure~\ref{fig:doubling} we can therefore uniquely identify the two strings $U'_6$ and $U'_7$.

In total,
around half of the substrings from $\calU$ in $\U[t_0,t_1]$ are uniquely identified through the outputs $\D[t_1+1,t_2]$. More precisely, $\Theta(2^\ell/r)$ substrings are uniquely identified. As the substrings were chosen uniformly at random from $\calU$, we have by Lemma~\ref{lem:combinatorial} that the entropy of $\D[t_1+1,t_2]$ is $\Omega(2^\ell/r\cdot r\log r)=\Omega(\log r\cdot 2^\ell)=\Omega(\delta\cdot 2^\ell)$ since $r$ was defined to be $2^\delta-1$. This property holds only for sufficiently large intervals, where our suggested choice of $n$ being at least $r^2$ comfortably makes the property true for intervals of length at least some constant times $\sqrt{n}$.

To sum it up, we have described a string $F$ and a set $\calU$ such that if the update sequence $U$ is picked by concatenating strings chosen independently and uniformly at random from $\calU$, we have the lower bound on the conditional entropy of $\D[t_1+1,t_2]$ of Lemma~\ref{lem:entropy-lower}.
%

\section{A string with many~different~Hamming~arrays}
\label{sec:R}
Our remaining task is now to show that a string $R$ does indeed exist which gives  many different Hamming arrays,  as demanded by Lemma~\ref{lem:combinatorial}. This is both the most important and the most technically detailed part of our overall lower bound proof.  To recap, we claim that for any $r$ there exists a string $R \in [r]^r$ which permits a large number of distinct Hamming arrays when compared to every string in $[r]^{2r}$, precisely, that there exists a string $R$ with $\log \big|\set{\Hamarray(R,U') | \textup{$U' \in [r]^{2r}$}}\big| \in \Omega(r\log r)$. 

\subsection{The structure of $R$}

The string $R$ is constructed by concatenating $\mu^2 = \Theta(r^{2/3})$ substrings each of length $\mu=\Theta(r^{1/3})$, containing exactly two symbols. One of the symbols will be common to all substrings and the other unique to that substring.  Denote the $i$-th such substring by $\ps_i$ and hence $R=\ps_0\ps_1\cdots\ps_{(\ELL^2-1)}$. Each substring of $R$, $\ps_i$, will correspond to a binary vector $v_i$ in the following natural way. Let $V=\{v_0,\dots,v_{(\ELL^2-1)}\}$ be a multi-set of $\ELL$~length vectors from $\{0,1\}^\ELL$ which will have the property set out in Lemma~\ref{lem:vecsum} below.  The string $\ps_i \in \{\star,i\}^\ELL$ is then given by taking $v_i$ and replacing every occurrence of $1$ with an occurrence of the symbol $i$ and similarly every $0$ with the symbol $\star$ (formally $\star=\mu^2 \in [r]$). For example, if $\ELL = 3$, $v_2 = (0,1,0)$ and $v_7= (1,1,0)$ then $\rho_2 = \star2\star$ and $\rho_7 = 77\star$. Observe that $\ps_i$ contains only two symbols and the symbol $i$ occurs only in $\ps_i$. We will assume w.l.o.g. that $R$ is a perfect cube $r=\mu^3$ and that $\mu-1$ is a prime which we will also require below. The result generalises to arbitrary $r$ via a simple reduction to a smaller $r$ which meets the assumptions.

The substrings $\rho_i$ of $R$ can be seen as encodings of binary vectors from a multi-set $V$ by the method described above.  We will also show that by suitably selecting the updates $U'$ from the unique symbols in $R$, the Hamming distances that result will be element-wise sums of the vectors from this multi-set.  We will therefore have reduced the problem of finding a string giving a large number of distinct Hamming arrays to that of finding multi-sets with a large number of distinct vector sums.  Lemma~\ref{lem:vecsum} captures the property that we require of $V$. The reason that the property must hold for large
multi-subsets\footnote{We use the term \emph{multi-subset} of $V$ to denote any multi-set obtained from $V$ by removing zero or more elements, e.g. $\{1,1,4,5,5\}$ is a multi-subset of $\{1,1,1,4,4,5,5,7,8\}$.}
of $V$ will become apparent from the construction below. Intuitively it is because we will `use up' vectors from $V$ as we proceed, and therefore we may assume that we always operate in the most pessimistic scenario where many vectors from $V$ are already unavailable.

In order to state Lemma~\ref{lem:vecsum} succintly, we need to define what we mean by distinct elements of a multi-set~$X$. We consider an arbitrary ordering of the elements of $X$ and refer to $X[i]$ as the $i$th element of $X$.  We say that elements $x_1,x_2, x_3,\ldots \in X$ are \emph{distinct} if $x _1= X[i_1], x_2 = X[i_2], x_3 = X[i_3], \dots$ and $i_j \neq i_{j'}$ for $j \neq j'$. In this way $x_1 = 1$, $x_2=1$, $x_3=3$ are distinct elements of the multi-set $\{1,1,2,3\}$, but $x_1 = 1$, $x_2=1$, $x_3= 1$ are not.

\begin{lemma}
    \label{lem:vecsum}
    For any $\ELL>40$ such that $\ELL-1$ is a prime, there exists a multi-set $V$ of vectors from
    $\{0,1\}^\ELL$ such that $|V|=\ELL^2$ and for any multi-subset $V'\subseteq V$ of size at least
    $(63/64)|V|$,
    \begin{align*}
        \left|\set{w_1 + \cdots + w_\ELL \,|\, \text{distinct }w_1,\dots,w_\ELL \in V'}\right| \geq \ELL^{(\ELL/10)} .
    \end{align*}

\end{lemma}

\subsection{Vector sums and Hamming arrays -- the proof of Lemma~\ref{lem:combinatorial}}

We can now prove Lemma~\ref{lem:combinatorial}. Our method is to show how to obtain a large number of Hamming arrays from the string $R$ by incrementally modifying an initial string of length $2r$ which $R$ will be compared to and which does not contain any symbol in $R$.   Let us call such a string $\U'=\{\tsymb\}^{2r}$, and let $\tsymb$ be a special symbol which does not occur in $R$ (formally $\tsymb=\mu^2+1 \in [r]$).

Consider the first alignment of $R$ and $\U'$ where $R[0]$ is aligned with $\U'[0]$ (refer to the top of Figure~\ref{fig:blocks}). From the set $V$, pick any $\ELL$ vectors and identify their corresponding substrings $\ps_i$ of $R$. For each such substring, set the symbol of $\U'$ that is directly to the right of $\ps_i$ in the alignment to $i$. For example, in Figure~\ref{fig:blocks}, where we have explicitly written out $\ps_0$, $\ps_5$ and $\ps_7$, suppose that $v_0$, $v_5$ and $v_7$ are among the $\ELL$ vectors we picked. As shown in the figure, we set three symbols of $\U'$ to 0, 5 and 7, accordingly.

\begin{figure*}[t]
    \centering
    \insertdiagram{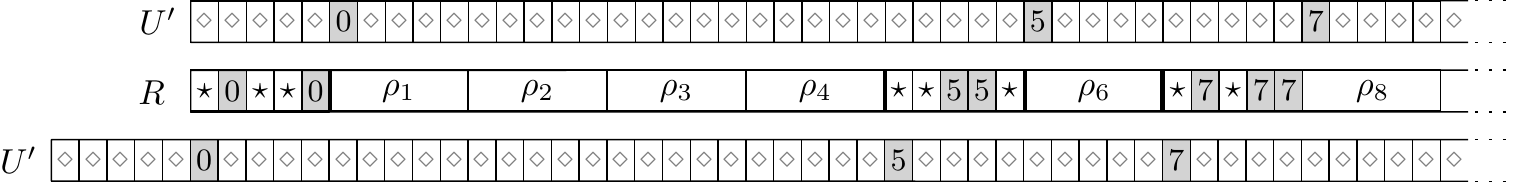}
    \caption{\label{fig:blocks}Setting symbols of $\U'$ renders a large set of possible Hamming distance outputs.}
\end{figure*}

Now consider the first $\ELL$ Hamming distances in $\Hamarray(R,U')$. We can think of these as being outputted as we slide the string $\U'$, $\ELL$ characters to the left (relative to $R$). The bottom part of the figure illustrates the alignment after sliding $\U'$. We see that the number of matches at each one of the $\ELL$ alignments correspond exactly to the vector sum of the $\ELL$ vectors we picked (in reverse order, to be precise).

We now repeat this process by picking $\ELL$ new vectors from $V$ and setting symbols of $\U'$ accordingly. We cannot pick a vector for which the position is already occupied by a symbol other than $\tsymb$. For the example in the figure, we would not be able to pick $v_4$ or $v_6$. If we did so, we would risk changing the previous Hamming distances. This is in fact the reason that Lemma~\ref{lem:vecsum} must hold not only for $V$ but also for large multi-subsets of $V$. The procedure of picking $\ELL$ vectors and sliding $\U'$ by $\ELL$ steps is repeated a total of $\ELL/64$ times, over which a total of $\ELL^2/64$ symbols of $\U'$ have been set. As $|V|=\ELL^2$, we have in each one of the $\ELL/64$ rounds had access to pick at least $(63/64)|V|$ vectors. Thus, by Lemma~\ref{lem:vecsum}, in each round we had a choice of at least $\ELL^{(\ELL/10)}$ distinct Hamming distance outputs. For correctness it is important to observe that setting a symbol of $\U'$ to some value $i$ only makes this symbol contribute to matches over exactly the $\ELL$ alignments it is intended for. For all other alignments the symbol will mismatch.

The process of performing $\ELL/64$ rounds as above is itself repeated $\ELL$ times. To see how this is possible, apply the following trick: slide $\U'$ left by one single step. By doing this, we offset all symbols, freeing up every occupied position of $\U'$ so that we can perform the process above again. The single-slide trick can be repeated $(\ELL-1)$ times, after which occupied positions will no longer necessarily be freed up but instead reoccupied by symbols that were set during the first rounds.

To sum up, we slide $\U'$ a total of $(\ELL\cdot(\ELL/64)+1)\cdot\ELL \leq \ELL^3/32=r/32$ steps. Over these steps, by Lemma~\ref{lem:vecsum} we have the choice of at least $(\ELL^{(\ELL/10)})^{(\ELL/64)\cdot\ELL}=\ELL^{(\ELL^3/640)}=\ELL^{(r/640)}$ Hamming array outputs. So we have $\log \big|\set{\Hamarray(R,U') | \textup{$U' \in [r]^{2r}$}}\big| \geq (r/640)\log\ELL\in\Omega(r\delta)$ since $\ELL \geq 2^{d/2-1}$. This completes the proof of Lemma~\ref{lem:combinatorial}.

\subsection{Vector sets with many distinct sums -- the proof of Lemma~\ref{lem:vecsum}}\label{sec:proofvecsum}

In this section, we prove Lemma~\ref{lem:vecsum}. We first rephrase it slightly for our purposes.  For any $V' \subset \{0,1\}^\ELL$, we define $\vsum(V') = \set{w_1 + \cdots + w_\ELL \,|\, \text{distinct }w_1,\dots,w_\ELL \in V'}$. Here vector addition is element-wise and over the integers.   We will show that  exists a multi-set $V$  of vectors from
    $\{0,1\}^\ELL$ such that $|V|=\ELL^2$ and for any multi-subset $V'\subseteq V$ of size at least $(63/64)|V|$, we have that $|\vsum(V')| \geq \ELL^{(\ELL/10)}$.

  Our approach will be an application of the probabilistic method. Specifically, we will show that when sampled uniformly at random, the expected value of $\min_{V'} |\vsum(V')| \geq \ELL^{(\ELL/10)}$ and hence there exists a $V$ with the required property.  To prove this result, we will require the following lemma from the field of Coding Theory, which is tailored for our needs and is a special case of ``Construction~II'' in~\cite{AGM:1992}. For our purposes, a binary constant-weight cyclic code can be seen simply as set of bit-strings (codewords) with two additional properties. The first is that all codewords have constant Hamming weight $\ELL$, i.e. they have exactly $\ELL$ ones. The second is that any cyclic shift of a codeword is also a codeword.

\begin{lemma}[\cite{AGM:1992}]
    \label{lem:cyclic-code}
    For any $\ELL\geq 4$ such that $\ELL-1$ is a prime and any odd $q\in[\ELL]$, there is a binary constant-weight cyclic code with $(\ELL-1)^{q}$ codewords of length $\ELL(\ELL-1)$ and Hamming weight $\ELL$ such that any two codewords have Hamming distance at least $2(\ELL-q)$.
\end{lemma}

    We will show that Lemma~\ref{lem:vecsum} holds for $V$ of size $(\ELL-1)\ELL<\ELL^2$.
    We first consider a random multi-set $V$ where the vectors are chosen independently and uniformly at random from $\{0,1\}^\ELL$. Later in the proof we will fix the multi-set $V$ and show that it has the property set out in the statement of the lemma.

    A multi-subset of $V$ of size $\ELL$ can be represented by a $|V|$~length bit string with Hamming
    weight~$\ELL$, where a 1 at position $i$ means that the $i$th vector of $V$ is in the multi-subset. Let
    $\Cbig$ be the binary code that contains all codewords of length $|V|$ with Hamming weight $\ELL$. That
    is, $\Cbig$ represents all $\ELL$-sized multi-subsets of $V$. To shorten notation, we refer to $\cbig\in\Cbig$ as both a codeword and a vector set.

    We now let $C\subseteq \Cbig$ be a smaller code such that $C$ is cyclic (i.e. $c[0]c[1]\cdots
    c[\ELL-1]\in C$ implies that also $c[1]c[2]\cdots c[\ELL-1]c[0]\in C$) and the Hamming distance between
    any two codewords in $C$ is at least $7\ELL/4$. We choose $C$ such that its size is $(\ELL-1)^q$, a value
    between $(\ELL-1)^{\ELL/9}$ and $(\ELL-1)^{\ELL/8}$, where $q$ is any odd integer in the interval
    $[\ELL/9,\ELL/8]$. The existence of such a $C$ is guaranteed by Lemma~\ref{lem:cyclic-code}. Like for
    $\Cbig$, every codeword of $C$ has Hamming weight $\ELL$.

    For $c\in C$, we define the \emph{ball}, $\ball{c}=\set{\cbig | \text{$\cbig\in\Cbig$ and
    $\Ham{c,\cbig}\leq \ELL/16$}}$ to be the set of bit strings in $\Cbig$ of weight $\ELL$ at Hamming distance at most
    $\ELL/16$ from $c$. Hence, the $|C|$ balls are all disjoint. We have that for any $c\in C$, using the fact ${a \choose b}\leq (ae/b)^b$,
    \begin{equation*}
        \label{eq:ball-size}
        \big|\ball{c}\big| \leq {\ELL \choose \ELL/16}\cdot{|V| \choose \ELL/16}
            \leq \left(\frac{\ELL e\cdot|V|e}{(\ELL/16)^2}\right)^{\ELL/16}
            \leq \left(\frac{\ELL}{16}\right)^{\ELL/16} .
    \end{equation*}

    For $\cbig\in \Cbig$, we write $\vsum(\cbig)$ to denote the vector in $[\ELL+1]^\ELL$ obtained by
    adding the $\ELL$ vectors in the vector set $\cbig$. For any $\cbig_1\in\ball{c_1}$ and
    $\cbig_2\in\ball{c_2}$, where $c_1,c_2\in C$ are distinct, we now analyse the probability that
    $\vsum(\cbig_1)=\vsum(\cbig_2)$. From the definitions above, it follows that $\cbig_1$ and $\cbig_2$ must differ
    on at least $7\ELL/4-2(\ELL/16)\geq \ELL$ positions, implying that the two vector sets $\cbig_1$ and
    $\cbig_2$ have at most $\ELL/2$ vectors in common, thus at least $\ELL/2$ of the vectors in $\cbig_1$ are not in $\cbig_2$. Let $v_1,\dots,v_{\ELL/2}$ denote those vectors.
    In order to have $\vsum(\cbig_1)=\vsum(\cbig_2)$, for each position $i\in [\ELL]$, the sum
    $s_i=v_1[i]+\cdots+v_{\ELL/2}[i]$ must be some specific value (that depends on the other vectors). Due
    to independence between vectors and their uniform distribution (any element of any vector is 1 with
    probability $1/2$), the most likely value of $s_i$ is $\ELL/4$ for which half of the vectors
    $v_1,\dots,v_{\ELL/2}$ have a 1 at position $i$.
    The probability of having $s_i=\ELL/4$ is exactly ${\ELL/2 \choose \ELL/4 }\cdot 2^{-\ELL/2} \leq
    (\ELL/2)^{-1/2}$, as for any $a$, ${a \choose a/2}\leq 2^a/\sqrt{a}$. Due to independence between the elements of a vector, the probability that all sums
    $s_0,\dots,s_{\ELL-1}$ combined yield $\vsum(\cbig_1)=\vsum(\cbig_2)$ is upper bounded by $(\ELL/2)^{-\ELL/2}$.
    Thus,
    \begin{equation}
        \label{eq:sum-equal}
        \Prob\big(\vsum(\cbig_1)=\vsum(\cbig_2)\big) \, \leq \, \left(\frac{\ELL}{2}\right)^{-\ELL/2} .
    \end{equation}

    For two distinct $c_1,c_2\in C$, we define the indicator random variable $I(c_1,c_2)$ to be 0 if and
    only if there exists a $\cbig_1\in\ball{c_1}$ and a $\cbig_2\in\ball{c_2}$ such that
    $\vsum(\cbig_1)=\vsum(\cbig_2)$. Taking the union bound over all $\cbig_1\in\ball{c_1}$ and
    $\cbig_2\in\ball{c_2}$, and using the probability bound in Equation~(\ref{eq:sum-equal}), we have
    \begin{align}
        \label{eq:union-inner}
        \Prob \big( I(c_1, c_2) = 0 \big)
        \, &\leq \, \big|\ball{c_1}\big|\cdot \big|\ball{c_2}\big| \cdot \left(\frac{\ELL}{2}\right)^{-\ELL/2}  \\
        \, &\leq \, \left(\frac{\ELL}{16}\right)^{2(\ELL/16)} \left(\frac{\ELL}{2}\right)^{-\ELL/2}
        \, \leq \, \left(\frac{1}{\ELL^3}\right)^{\ELL/8} \,. \notag
    \end{align}

    For any $c_1\in C$, we now define the indicator random variable $I'(c_1)$ to be 0 if and only if there
    exists some $c_2\in C\setminus \{c_1\}$ such that $I(c_1,c_2)=0$. Taking the union bound over all $c_2\in
    C$ and using Equation~(\ref{eq:union-inner}), we have
    \begin{align}
        \label{eq:union-outer}
        \Prob \big( I'(c_1) = 0 \big) \, &\leq \!\!\! \sum_{c_2 \in C\setminus \{c_1\}} \!\!\!\! \Prob \big( I(c_1, c_2) = 0 \big)\\
        \, &\leq \, |C| \left(\frac{1}{\ELL^3}\right)^{\ELL/8}
        \, \leq \, \ELL^{(\ELL/8)} \left(\frac{1}{\ELL^3}\right)^{\ELL/8}
        \, \leq \, \frac{1}{2} \,. \notag
    \end{align}

    We say that $\ball{c_1}$ is \emph{good} iff $I'(c_1)=1$. From the definitions above we have that for every $\cbig_1$ in a
    good ball, there is no other ball that contains a $\cbig_2$ such that $\vsum(\cbig_1)=\vsum(\cbig_2)$. It is
    possible that $\vsum(\cbig_1)=\vsum(\cbig_2)$ if $\cbig_2$ is from the same ball as $\cbig_1$ though. The
    expected number of good balls is, by linearity of expectation and Equation~(\ref{eq:union-outer}),
    $\expected{\sum_{c\in C}I'(c)} \geq |C|/2$. The conclusion is that there is a multi-set $V$ of vectors for
    which at least $|C|/2$ balls are good, hence $\vsum(V)\geq (\ELL-1)^{\ELL/9}/2$. From now on, we fix $V$ to be such a multi-set. It remains to show that for any multi-subset $V'$ of $V$ of size $(63/64)|V|$, $\vsum(V')$ is also large.

    Over all codewords in $C$, the total number of 1s is $|C|\ELL$. As $C$ is cyclic, the number of
    codewords that have a 1 in position $i\in [|V|]$ is the same as the number of codewords that have a 1 in any position
    $j\neq i$. Thus, for each one of the $|V|$ positions there are exactly $|C|\ELL/|V|$ codewords in $C$
    with a 1 in that position.

    Let $V'$ be any multi-subset of $V$ of size $(63/64)|V|$. Let $J$ be the set of $|V|/64$ positions that correspond to the vectors of $V$ that are not in $V'$.
    For each $j\in J$ and codeword $c\in C$, we
    set $c[j]$ to 0. The total number of 1s is therefore reduced by exactly $(|V|/64) \cdot
    (|C|\ELL/|V|) = |C|\ELL/64$. The number of codewords of $C$ that have lost $\ELL/16$ or more 1s is
    therefore at most $(|C|\ELL/64)/(\ELL/16)=|C|/4$. Let $C'\subseteq C$ be the set of codewords $c$
    that have lost less than $\ELL/16$ 1s and for which $\ball{c}$ is good. As there are at least $|C|/2$
    good balls, $|C'|\geq |C|/4$. Let the code $C''$ be obtained from $C'$ by replacing, for each codeword $c'\in C'$, every removed 1 with a 1 at some other arbitrary position that is not in $J$. Thus, every codeword of
    $C''$ has Hamming weight $\ELL$ and they all belong to $|C''|=|C'|\geq |C|/4$ distinct good balls.
    Further, every codeword of $C''$, seen as a vector set, only contains vectors from the subset $V'$.
    From the definition of a good ball we have that at least $|C|/4$ distinct vector sums can be
    obtained by adding $\ELL$ vectors from $V'$. Thus, $\vsum(V')\geq (\ELL-1)^{\ELL/9}/4\geq \ELL^{(\ELL/10)}$
    when $\ELL>40$. This completes the proof of Lemma~\ref{lem:vecsum}.

\section*{Acknowledgements}
RC would like to thank Elad Verbin, Kasper Green Larsen, Qin Zhang and the members of CTIC for helpful and insightful discussions about lower bounds during a visit to Aarhus University. We thank Kasper Green Larsen in particular for pointing out that the cell-probe lower bounds we give are in fact tight. We also thank the anonymous reviewers for their helpful comments. Some of the work on this paper has been carried out during RC's visit at the University of Washington.


\printbibliography

\newpage

\appendix

\section{Folklore matrix multiplication reduction}

We show a reduction\footnote{This reduction is attributed to Ely Porat according to Rapha\"{e}l Clifford. Ely Porat attributes it to Piotr Indyk. Piotr Indyk denies this.} from binary matrix multiplication to pattern matching under the Hamming distance.

Consider the following reduction. Assume the input is of two binary matrices $A$ and $B$ of sizes $m \times \ell$ and $\ell \times n$.  For matrix $A$, we write $x$ for each $0$ and for each $1$ we write its column number. For example, $A=((0,0,1),(1,0,1))$ is translated to $A'=((x, x, 3),(1, x, 3))$.  For matrix $B$, we write $y$ for each $0$ and the row number for each $1$.   For example, $B=(0,1),(1,0), (0,0))$ is translated to $B'=((y,1),(2,y),(y,y))$.  Now create pattern $P$ as the concatenation of the rows of $A'$ and text $T$ as the concatenation of the columns of $B'$ with the unique symbol \$  inserted after every column and add $m(\ell-1)$ \$ symbols at the beginning and end of $T$. So, in our example $P = xx31x3$ and $T=\$\$\$y2y\$12y\$\$\$$.

We now count the number of matches between $P$ and $T$ at each alignment, giving in this case $0,0,0,0, 1, 0, 0, 0$ meaning that the second row of $A$ scored $1$ when multiplied with the second column of $B$.  The trick is that the \$ symbols force at most one substring of the pattern corresponding to a row in $A$ to match one substring of $T$ corresponding to a column of $B$ at any given alignment.

\end{document}

%